\documentclass[letterpaper, 10 pt, conference]{ieeeconf}  

\IEEEoverridecommandlockouts                              

\newcommand{\R}{\mathbb{R}}
\newcommand{\N}{\mathbb{N}}

\usepackage{float}
\usepackage{amssymb}
\usepackage{amsmath}
\usepackage{bbold}
\usepackage{url}          
\usepackage{hyperref}
\usepackage{caption}
\usepackage{comment}
\usepackage[ruled]{algorithm2e}
\usepackage{float}
\usepackage{multicol}
\usepackage{color}
\usepackage{graphicx}
\newcounter{multifig}
\usepackage{thmtools,thm-restate}
\usepackage{multirow}

\newcommand\norm[1]{\left\lVert #1 \right\rVert} 

\usepackage{lipsum}

\usepackage[caption=false, font=footnotesize]{subfig}

\newcommand{\figcaption}[1]
{\stepcounter{multifig}
	\addcontentsline{lof}{figure}{\string\numberline {\arabic{multifig}}{\ignorespaces #1}}
	Figure \arabic{multifig}: #1}

\newtheorem{theorem}{Theorem}[section]

\newtheorem{definition}[theorem]{Definition}
\newtheorem{remark}[theorem]{Remark}

\newtheorem{assumption}[theorem]{Assumption}

\title{\LARGE \bf On the Control of Agents Coupled through Shared Unit-demand Resources}

\author{Syed Eqbal Alam$^\ast$\thanks{$^\ast$Concordia Institute for Information Systems Engineering, Concordia University, Montreal, Quebec, Canada}, 
	Robert Shorten$^\dagger$\thanks{$^\dagger$School of Electrical, Electronic
		and Communications Engineering, University College Dublin, Dublin, Ireland},
	Fabian Wirth$^\ddagger$\thanks{$^\ddagger$Faculty of Computer Science and Mathematics, University of Passau, Passau, Germany}, and
	Jia Yuan Yu$^\ast$}

\begin{document}
	
	\maketitle
	\thispagestyle{empty}
	\pagestyle{empty}
	
	\begin{abstract}
	    We consider a control problem involving several agents coupled through multiple unit-demand resources. Such resources are indivisible, and each agent's consumption is modeled as a Bernoulli random variable. Controlling the number of such agents in a probabilistic manner, subject to capacity constraints, is ubiquitous in smart cities. For instance, such agents can be humans in a feedback loop---who respond to a price signal, or automated decision-support systems that strive toward system-level goals. In this paper, we consider both single feedback loop corresponding to a single resource and multiple coupled feedback loops corresponding to multiple resources consumed by the same population of agents. For example, when a network of devices allocates resources to deliver several services, these services are coupled through capacity constraints on the resources. We propose a new algorithm with fundamental guarantees of convergence and optimality, as well as present an example illustrating its performance.
	\end{abstract} 
	\let\thefootnote\relax\footnotetext{The  work  is  partly  supported  by  Natural  Sciences  and  Engineering  Research  Council of Canada grant no. RGPIN-2018-05096.}	
	
	
	\section{Introduction}
		Classical control has much to offer in a smart-city context. However, while this is, without doubt, true, many problems arising in the context of smart cities reveal subtle constraints that are relatively unexplored by the control community. At a high level, both classical control and smart-city control deal with regulation problems. Nevertheless, in many (perhaps most) smart-city applications, control involves orchestrating the aggregate effect of a number of agents who respond to a signal (sometimes called a {\em price}) in a probabilistic way. A fundamental difference between classical control and smart-city control is the need to study the effect of control signals on the statistical properties of the populations that we wish to influence, while at the same time ensuring that the control is in some sense optimal. This fundamental difference concerns the need of ergodic feedback systems, and even though this problem is rarely studied in control, it is the issue that is perhaps the most pressing in real-life applications; since the need for predictability at the level of individual agents underpins an operator's ability to write economic contracts.
	
	Our starting point for this work is the previous papers \cite{Griggs2016, Fioravanti2017}, and the observation that many problems in smart cities can be cast in a framework, where a large number of agents, such as people, cars, or machines, often with unknown objectives---compete for a limited resource. It is a challenge to allocate a resource in a manner that utilizes it optimally and gives a guaranteed level of service to each of the agents competing for that resource. For example, allocating parking spaces \cite{Arnott1999, Teodorovic2006, Lin2017}, regulating cars competing for shared road space \cite{Jones2014}, or allocating shared bikes \cite{Raviv2013, DeMaio2009}, are examples in which resource utilization should be maximized, while at the same time delivering a certain quality of service (QoS) to individual agents is a paramount constraint.  As we have noted in \cite{Fioravanti2017, Syed2018_B, Crisostomi2016}, at a high level, these are primarily optimal control problems but with the added objective of controlling the microscopic properties of the agent population. Thus, the design of feedback systems for deployment in smart cities must combine notions of regulation, optimization, and the existence of this unique invariant measure \cite{Fioravanti2019}.
	
	Specifically, in this paper, we consider the problem of controlling a number of agents coupled through multiple shared resources, where each agent demands the resources in a probabilistic manner. This work builds strongly on the previous work \cite{Griggs2016} in which the optimal control and ergodic control of a single population of agents are considered. As we have mentioned, controlling   networks of agents which demand resources in a {\em probabilistic manner} is ubiquitous in smart cities. In many smart-city applications, the probabilistic intent of agents can be natural (where humans are in a feedback loop and respond, for example, to a price signal), or designed (implemented in a decision support system) so that the network achieves system-level goals. Often, such feedback loops are coupled together, as agents contribute or participate in multiple services. For example, when a network of devices allocates resources to deliver several services, these services are coupled through the consumption of multiple shared resources; usually, we call such resources as {\em unit-demand resources} which are either allocated one unit of the resource or not allocated. A concrete manifestation of such a system is the IBM Research's project {\em parked cars as a service delivery platform} \cite{Cogill2014}. Here, networks of parked cars collaborate to offer services to city managers. Examples of services include {\em wifi coverage}, {\em finding missing objects}, and {\em gas leak detection and localization}.  Here, vehicle owners allocate parts of their resource stochastically to contribute to different services, each of which are managed by a feedback loop. The allocation between services is usually coupled via a nonlinear function that represents the trade-off between resource allocation (energy, sensors), and the reward for participating in delivering a particular service. We shall give a concrete example of such a system later in the paper. It is our firm belief that such systems are ubiquitous in smart cities, and represent a new class of problems in feedback systems.
	
	Our main contribution in this paper is to establish stochastic schemes for a practically important class of problems for several agents coupled through multiple unit-demand shared resources in coupled feedback loops. Each agent demands the unit-demand shared resources in a probabilistic manner based on its private cost function and constraints; the constraints are based on multiple unit-demand shared resources. This scheme is a generalization of the single unit-demand resource allocation algorithm proposed in \cite{Griggs2016} and follows more relaxed constraints than \cite{Syed2018_B}. Furthermore, the results of convergence, as well as optimality, are derived for networks with a single unit-demand resource; the results are further extended for multiple unit-demand resources. 
   	%
%
	\section{Preliminaries} \label{prob_form}
	Suppose that $n$ agents are coupled through $m$ resources $R^1, R^2, \ldots, R^m$ and each agent has a cost function that depends on the allocation of these resources in the closed coupled feedback loop. Let the desired value or capacity of resource $R^1, R^2, \ldots, R^m$ be $C^1, C^2, \ldots, C^m$, respectively. We denote $\mathcal{N} \triangleq \{1, 2, \ldots, n\}$, $\mathcal{M} \triangleq \{1, 2, \ldots, m\}$, and use $i \in \mathcal{N}$ as an index for agents and $j \in \mathcal{M}$ to index the resources. 
	Let $\xi^j_i(k)$ denote independent Bernoulli random variable which represents the instantaneous allocation of resource $R^j$ of agent $i$ at time step $k$. Furthermore, let $y_i^j(k) \in [0, 1]$ denote the average allocation of resource $R^j$ of agent $i$ at time step $k$. We define $y_i^j(k)$ as follows,
\begin{align} \label{eq:average_eqn_b1}
{y}_i^j(k) \triangleq \frac{1}{k+1} \sum_{\ell=0}^k \xi^j_i(\ell),
\end{align}
for $i = 1, 2, \ldots, n$, and $j = 1, 2, \ldots, m$. We assume that agent $i$ has a cost function $g_i: (0,1]^m \to \mathbb R_+$ which associates a cost to a
certain allotment of resources to the agent. We assume that $g_i$ is twice continuously differentiable, convex, and increasing in all variables, for all $i$. We also assume that the agents do not share their cost functions or allocation information with other agents. Then instead of defining the resource allocation problem in terms of the instantaneous allocation $\xi_i^j(k) \in \{0, 1\}$, for all $i,j$ and $k$, we define the objective and constraints in terms of averages as follows,
\begin{align}
\begin{split} \label{obj_fn1_b1}
\min_{y_1^1, \ldots, y_n^m} \quad &\sum_{i=1}^{n} g_i(y_i^1,\ldots,y_i^m),
\\ \mbox{subject to } \quad &\sum_{i=1}^{n} y_i^j  =  C^j, \quad  j=1,\ldots,m, 
\\  &y_i^j\geq 0, \quad i=1,\ldots,n, \  \text{ and}, \ j=1,\ldots,m.
\end{split}
\end{align}

Let $y^* = ({y}_1^{*1}, \ldots, {y}_n^{*m}) \in (0,1]^{nm}$ denote the solution to \eqref{obj_fn1_b1}. Let $\mathbb{N}$ denote the set of natural numbers, and let $k \in \mathbb{N}$ denote the time steps.
Next, our objective is to propose a distributed iterative algorithm that determines instantaneous allocation $\{\xi_i^j(k)\}$ and ensures that the long-term average allocation, as defined in \eqref{eq:average_eqn_b1} converge to optimal allocation as follows (treated in Section~\ref{bin_imp}), 
\begin{align*}
\lim_{k\to \infty} {y}_i^j(k) = {y}_i^{*j}, \text{ for } i=1,2,\ldots,n, \text{ and } j=1,2,\ldots,m,
\end{align*}
 thereby achieving the
minimum social cost in the sense of long-term averages.
By compactness of the constraint set, optimal
solutions exist. The assumption that the cost function $g_i$ is strictly convex leads to strict convexity of $\sum_{i=1}^{n} g_i$, which follows that the optimal solution is unique. 
	\subsection{Optimality conditions} \label{opt_cond_b1}
	Let $\mathcal{L}: \mathbb{R}^{nm}\times \mathbb{R}^m \times \mathbb{R}^m \to \mathbb{R}$, and let $\mu = (\mu^1,\mu^2,\ldots,\mu^m)$ and $\lambda = (\lambda^1,\lambda^2,\ldots,\lambda^m)$ are Lagrange multipliers of the resources. Then we define Lagrangian of Problem \eqref{obj_fn1_b1} as follows, 
	\begin{align*}
		&\mathcal{L}(y, \mu, \lambda) \triangleq \nonumber\\ & \sum_{i=1}^{n} g_i(y_i^1,
		\ldots, y_i^m) -\sum_{j=1}^{m}\mu^j (\sum_{i=1}^{n} y_i^j - C^j) + \sum_{j=1}^{m} \sum_{i=1}^{n} \lambda^j y_i^j.
	\end{align*} 
	Recall that $y_i^{*1}, \ldots, y_i^{*m} \in (0,1]$ are the optimal allocations of agent $i$ of Problem \eqref{obj_fn1_b1}, for $i = 1,2, \ldots, n$.  
	 Now, let $\nabla_j g_i$ denote (partial) derivative of the cost function $g_i$ with respect to resource $R^j$, for $j=1,2,\ldots,m$. Then following similar analysis as \cite{Syed2018_B}, we find that the derivatives of the cost functions of all agents competing for a particular resource reach consensus at optimal average allocations. That is, the following holds true, for $i,u \in \mathcal{N}, \mbox{ and } j \in \mathcal{M}$:
	\begin{align}\label{optimality_cond_b1}
		\nabla_j g_i(y_i^{*1}, \ldots, y_i^{*m}) = \nabla_j g_u(y_u^{*1}, \ldots, y_u^{*m}).
	\end{align} 
	Furthermore, Karush-Kuhn-Tucker (KKT) conditions are satisfied by the consensus of derivatives (cf. \eqref{optimality_cond_b1}) of the cost functions that are necessary and sufficient conditions of optimality of the optimization Problem \eqref{obj_fn1_b1}; a similar analysis is done in \cite{Syed2018_B, Wirth2014, Syed2018}, readers can find further details of KKT conditions at Chap. 5.5.3 \cite{Boyd2004}.
	 In this paper, we use this principle to show that the proposed algorithm reaches optimal values asymptotically. The consensus of derivatives of cost functions are also used in \cite{Wirth2014, Nabati2018, Chaturvedi2018} (single resource), \cite{Syed2018_B, Syed2018} (multi-resource---stochastic), \cite{Syed2018_al} (multi-resource---derandomized) to show the convergence of allocations to optimal values.
\section{Allocating single unit-demand resource through a feedback loop} \label{prelim} 
In this section, we consider the single resource case of \cite{Griggs2016} and briefly describe the proposed distributed, iterative and stochastic allocation algorithm. We also provide proof of its convergence and optimality properties with a few assumptions.

With a single resource, we can simplify notation by dropping the index $j$. Each agent $i$ has a strictly
convex cost function $g_i:(0,1] \to \mathbb{R}_+$. The binary
random variable $\xi_i(k)
\in \{ 0, 1\} $ denotes the allocation of the unit
resource for agent $i$ at time step $k$. Let $y_i(k)$ be the
average allocation up to time step $k$ of agent $i$ that is, $y_i(k)
\triangleq \frac{1}{k+1}\sum_{\ell=0}^{k} \xi_i(\ell)$. Let
$\xi(k) \in \{0,1\}^n$ and $y(k) \in [0,1]^n$ denote the vectors with entries $\xi_i(k),y_i(k)$, respectively, for $i=1,2,\ldots,n$. 

The idea is to choose the probability for random variable $\xi_i$ so as to ensure
convergence to the socially optimum value and to adjust overall
consumption to the desired level $C$ by applying a normalization
factor $\Omega$ to the probability, for all $i$. When an agent joins the
network at time step $k \in \mathbb{N}$, it receives the normalization factor
$\Omega(k)$. At each time step $k$, the central agency updates $\Omega(k)$ using a gain parameter $\tau$, past utilization of the resource, and its
capacity; then it broadcasts the new value to all agents in the network, 
\begin{align} \label{omega_bs1}
\begin{split}
\Omega(k+1) \triangleq \Omega(k) -  \tau  \Big (\sum_{i=1}^n \xi_i(k) -C \Big),
\end{split}
\end{align}
\begin{align} \label{tau_bs1}
\text{where } \tau \in \Big( 0,  \Big( \max_{y \in
	[0,1]^n} \sum_{i=1}^{n} \frac{ y_i}{g_i'({y}_i)}  \Big)^{-1} \Big).
\end{align}
After receiving this signal, agent $i$ responds in a random fashion
based on its available information. The probability function
$\sigma_i(\cdot)$ uses the average allocation of the resource to agent $i$ and the derivative $g_i'$ of the cost function $g_i$, is given by,
\begin{align} \label{sigma_bs}
\sigma_i(\Omega(k),y_i(k)) \triangleq \Omega(k)
\frac{{y}_i(k)}{
	{g_i'({y}_i(k))}}, \text{ for } i \in \mathcal{N}.	
\end{align} 
Agent $i$ updates its resource demand at each time step either by demanding one unit of the resource or not demanding it, as follows,
\begin{align*}
\xi_i(k+1) =
\begin{cases} 
1 \quad \text{with probability }  \sigma_i(\Omega(k),y_i(k));\\ 
0 \quad \text{with probability }  1-\sigma_i(\Omega(k),y_i(k)).
\end{cases}
\end{align*}
We point out that for the above formulation we require
assumptions on the cost function $g_i$ and the admissible value of $\Omega$ because the scheme requires that \eqref{sigma_bs} does, in fact, define a probability. For ease of notation, we define $v_i(z) \triangleq z/g_i'(z)$, where $z \in [0,1]$, and $v(y)$ to be the vector with components $v_i(y_i)$, for $i=1,2,\ldots,n$, where $y \in [0,1]^n$.

\begin{definition}[Admissibility]
	\label{ass:algorithm}
	Let $n \in \N$, and let $g_i:[0,1] \to \R_+$ be continuously
	differentiable and strictly convex, for $i=1,\ldots,n$. We call the set $\{
	g_i, i=1,\ldots,n \}$ and $\Omega>0$ admissible, if 
	\begin{enumerate}
		\item[(i)] $v_i$ is well defined on $[0,1]$, for $i=1,\ldots,n$,
		\item[(ii)] there are constants $0<a<b<1$, such that
		$\sigma_i(\Omega,z)=\Omega v_i(z) \in [a,b]$, for $i=1,\ldots,n$, and $z\in [0,1]$.
	\end{enumerate}
\end{definition}
The definition of admissibility imposes several restrictions on the
possible cost function $g_i$, similar to those imposed in \cite{Wirth2014}. 
See this reference for a detailed discussion and possible relaxations.
For the case that $\Omega$ is a constant that is, $\Omega$ does not depend on time step $k \in \mathbb{N}$; therefore, \eqref{omega_bs1} is not active, the convergence of the scheme follows using tools from classical stochastic approximation \cite{Borkar2008}.
\begin{theorem} \cite[Theorem~2.2]{Borkar2008} 
If $x(k) \in \mathbb{R}_+^n$, for $k \in \mathbb{N}$, be formulated as follows
\begin{align} \label{eq:avg_y2}
	&x(k+1) = x(k) + a(k) \big [ h(y(k)) + M(k+1) \big],
	\end{align} 
for a fixed $x(0)$ and Assumptions \ref{as:Borkar} (i) to (iv) are satisfied; then $\{x(k)\}$ converges to a connected chain-transitive set of the differential equation 
	$\dot x(t) = h(x(t)), \text{ almost surely}, \text{ for } t \geq 0.$
	\end{theorem}
\begin{assumption} \label{as:Borkar}
\begin{itemize}
\item[(i)] The map $h$ is Lipschitz.
\item[(ii)] Step-size $a(k)>0$, for $k \in \mathbb{N}$, and 
\begin{align*}
\sum_{\ell=0}^{\infty} a(\ell) = \infty, \text{ and } \sum_{\ell=0}^{\infty} \big(a(\ell) \big)^2 < \infty.
\end{align*}
\item[(iii)] $\{M(k)\}$ is a martingale difference sequence with respect to the $\sigma$-algebra $\mathcal{F}_k$ generated by the events up to time step $k$. Also, for $l^2$-norm $\norm{\cdot}^2$, martingale difference sequence $\{ M(k)\}$ is square-integrable that is,
\begin{align*}
	\mathbb{E} \big(\norm{M(k+1)}^2 \vert {\cal F}_k \big ) \leq \eta(1+\norm{x(k)}^2), \\ \text{ almost surely, for } k \in \mathbb{N}, \text{ and } \eta >0 \nonumber.
\end{align*}
\item[(iv)] Sequence $\{x(k)\}$ is almost surely bounded. 
\end{itemize}
\end{assumption}	
Theorem on the convergence of average allocation $y(k)$ is stated as follows.
\begin{theorem}[Convergence of average allocations]
	\label{t:single-res}
	Let $n\in \mathbb{N}$. Assume that the cost function $g_i: [0,1] \to
	\R_+$ is strictly convex, continuously differentiable and strictly increasing in each variable, for $i=1,\ldots,n$. Let
	$\Omega>0$, and assume that $\{
	g_i, i=1,\ldots,n \}$ and $\Omega$ are admissible.
	Then almost surely, $\lim_{k\to\infty} y(k) = y^*$, where $y^*$ is
	characterized by the condition,
	\begin{equation}
	\label{eq:fixchar}
	\Omega = g'_i(y_i^*) , \quad \text{for } i=1,\ldots,n.
	\end{equation}
\end{theorem}
\begin{proof}
	By definition, we have,
	\begin{equation} \label{eq:avg_y0}
	y(k+1) = \frac{k}{k+1} y(k) + \frac{1}{k+1} \xi(k+1).
	\end{equation}
	Let $\sigma(\Omega, y(k))$ denote the vectors with entries $\sigma_i(\Omega, y_i(k))$, for $i=1,2,\ldots,n$, and $k =0, 1, 2, \ldots$ Thus, \eqref{eq:avg_y0} may be reformulated as
	\begin{align} \label{eq:avg_y}
	&y(k+1) = y(k) + \nonumber \\ & \hspace{0.2in}\frac{1}{k+1} \big [ \big( \sigma(\Omega, y(k)) - y(k) \big) + \big( \xi(k+1) - \sigma(\Omega, y(k)) \big) \big].
	\end{align} 
	Furthermore, let $\big( \xi(k+1) - \sigma(\Omega, y(k)) \big)$ be denoted by $M(k+1)$, and the step-size $\frac{1}{k+1}$ be denoted by $a(k)$, for $k \in \mathbb{N}$. Also, let $\big( \sigma(\Omega, y(k)) - y(k) \big)$ be denoted by $h(y(k))$. Then we can reformulate \eqref{eq:avg_y} similar to \eqref{eq:avg_y2}. 
	
	We can verify that Assumption \ref{as:Borkar} (i) to (iv) are satisfied for formulation \eqref{eq:avg_y}. Recall that $h(y(k)) = \big( \sigma(\Omega, y(k)) - y(k) \big)$; thus, the map $h: y \mapsto \sigma(\Omega, y) - y = \Omega v(y) -y$ is Lipschitz, which satisfies Assumption \ref{as:Borkar} (i). Also, the step-size $a(k)=\frac{1}{k+1}$ is positive, for $k=0,1,2, \ldots$, and we can derive that
		 \begin{align*}
\sum_{\ell=0}^{\infty} a(\ell) = \infty, \text{ and } \sum_{\ell=0}^{\infty} \big(a(\ell) \big)^2 < \infty,
\end{align*}
 which satisfy Assumption \ref{as:Borkar} (ii). Additionally, we note that the expectation:
	\begin{equation} \label{eq:martingale}
	\mathbb{E} \big( \xi(k+1) - \sigma (\Omega, y(k)) \vert {\cal F}_k \big ) = 0,
	\end{equation}
	where ${\cal F}_k$ is the $\sigma$-algebra generated by the events up to time step $k$. This follows immediately from the definition of the probability
	$\sigma_i(\cdot)$. By \eqref{eq:martingale}, we say that $\{M(k)\}$ is a martingale difference sequence with respect to $\sigma$-algebra; also, the sequence $\{ \xi(k+1) - \sigma(\Omega, y(k)) \} $ is of course bounded, with little manipulation we can show that the martingale difference sequence $\{M(k)\}$ is square-integrable---which satisfy Assumption \ref{as:Borkar} (iii). Moreover, the iterate $y(k) \in [0,1]^n$ is bounded almost surely, which satisfies Assumption \ref{as:Borkar} (iv). Thus, it follows that almost surely $\{y(k)\}$ converges to a connected chain-transitive set of the differential equation,
	\begin{equation}  \label{eq:dif}
	\dot y = \Omega v(y) - y.
	\end{equation}
	
	It remains to show that the differential equation has an asymptotically stable fixed point whose domain of attraction contains the set $[0,1]^n$, as this then determines the unique possible limit point of $\{
	y(k) \}$. We note first that the differential equation is given by $n$ decoupled equations
	\begin{equation*}
	\dot y_i = \Omega v_i(y_i) - y_i, \quad \text{for } i =1,2,\ldots,n.   
	\end{equation*}
	The fixed points for each of these $1$-dimensional equations are
	characterized by the condition $\Omega y_i^*/g'_i(y_i^*) - y_i^* =0$. We
	have by Definition~\ref{ass:algorithm}\,(ii) that 
	\begin{equation}
	\label{eq:uniquefix}
	\Omega v_i(0)>0  \quad \text{and}\quad \Omega v_i(1) -1 <0, \quad \text{for } i =1,2,\ldots,n.
	\end{equation}
	This shows that $y_i^* \in (0,1)$ and so a little manipulation shows that fixed points are characterized by,
	\begin{equation*}
	\Omega = g_i'(y_i^*), \quad \text{for } i = 1, 2, \ldots, n.
	\end{equation*}
	As $g_i$ is strictly convex, $g_i'$ is strictly increasing and so the fixed point for each of the decoupled equations is unique. Now,
	\eqref{eq:uniquefix} together with sign considerations shows asymptotic stability  and the desired property of the domain of attraction. The proof is complete.
\end{proof}
 Notice that proof of convergence is based on the constant $\Omega$; it is an open problem to prove convergence of average allocation with $\Omega(k)$ that varies with time step $k$ (cf. \eqref{omega_bs1}).
\begin{remark}[Optimality] \label{re_Optim_s}
	We note that the fixed point condition \eqref{eq:fixchar} can be
	interpreted as an optimality condition---as established in \eqref{optimality_cond_b1}. If we define $C^* \triangleq \sum_{i=1}^n
	y_i^*$ then \eqref{eq:fixchar} shows that $y^*$ is the unique optimal point of the optimization problem:
	\begin{align*}
	\min_{y_1, \ldots, y_n} & \sum_{i=1}^n g_i(y_i), \\ \text { subject to }
	& \sum_{i=1}^n y_i = C^*,\quad y_i \geq 0, \quad \text{for } i =1,2,\ldots,n.   
	\end{align*}
	Furthermore, the equation shows that $\Omega$ may be used to adjust the fixed points; thus the constraints. As the cost function $g_i$ is strictly convex and increasing in each variable, the derivative $g_i'$ is positive and increasing. Therefore, increasing $\Omega$ increases each $y_i^*(\Omega)$; thus the total constraint $C^*(\Omega)$, while decreasing $\Omega$ has the opposite effect. The simple PI controller for $\Omega$ in \eqref{omega_bs1} thus has the purpose of adjusting to the right level of resource consumption. The full proof of convergence of the scheme with PI-controller in the loop is beyond the scope of the present paper.
\end{remark}

	\section{Allocating multiple unit-demand resources through coupled feedback loops} \label{bin_imp}

We turn our attention in this section to the case of multiple resources shared by the same population of agents. We present a new algorithm that generalizes the single-resource algorithm of the previous section to multiple unit-demand resources. The agents are coupled through these shared resources.

	Before presenting the algorithm, we introduce the following additional notions. Suppose that there exists $\delta > 0$, such that $\mathcal G_{\delta}$ is a set of continuously differentiable, convex and increasing functions, and $g_1, g_2, \ldots, g_n \in \mathcal G_{\delta}$. We assume that $\mathcal G_{\delta}$ is common knowledge to the control unit, and each cost function $g_i$ is private and should be kept private. Although $\mathcal G_{\delta}$ is common knowledge, due to the large number of the cost functions $g_1,g_2, \ldots, g_n$ in $\mathcal G_{\delta}$, it is difficult for the control unit to guess the cost function $g_i$ of a particular agent $i$; it is true for every agent in the network. 

	Each agent in the network runs the distributed unit-demand multi-resource allocation algorithm. Let $\tau^j$ be the gain parameter, $\Omega^j(k)$ denotes the normalization factor (signal of the controller) of the feedback loop, and let $C^j$ represent the desired value (capacity) of resource $R^j$, respectively, for all $j$. We use the term {\em control unit} instead of controller here. The control unit updates $\Omega^j(k)$ according to \eqref{omega_b1} at each time step and broadcasts it to all agents in the network, for all $j$ and $k$. When an agent joins the network at time step $k$, it receives the parameter $\Omega^j(k)$ for resource $R^j$, for all $j$. Every agent's algorithm updates its resource demand at each time step---either by demanding one unit of the resource or not demanding it.
	The normalization factor $\Omega^j(k)$ depends on its value at the previous time step, $\tau^j$, capacity $C^j$, and the total utilization of resource $R^j$ at the previous time step, for all $j$ and $k$. After receiving this signal,  agent $i$'s algorithm responds in a probabilistic manner. It calculates its probability $\sigma^j_i(k)$ using its average allocation $y_i^j(k)$ of resource $R^j$ and the derivative of its cost function, for all $j$ and $k$, as described in \eqref{prob_x2}. Agent $i$ finds out the outcome of Bernoulli trial for resource $R^j$, outcome $1$ with probability $\sigma^j_i(k)$ and outcome $0$ with probability $1-\sigma^j_i(k)$; based on the value $0$ or $1$, the algorithm decides whether to demand one unit of the resource $R^j$ or not. If the value is $1$, then the algorithm demands one unit of the resource; otherwise, it does not demand the resource, analogously, it is done for all the resources. This process repeats over time. We present the proposed {\em unit-demand multi-resource allocation} algorithm for the control unit in Algorithm \ref{algoCU2} and the algorithm for each agent in Algorithm \ref{algo3}.
	
	\begin{algorithm}  \SetAlgoLined Input:
		$C^{1}, \ldots, C^{m}$, $\tau^1, \ldots,\tau^m$, $ \xi_i^1(k), \ldots, \xi_i^m(k)$, for $k \in \mathbb{N}$ and $i \in \mathcal{N}$.
		
		Output:
		$\Omega^{1}(k+1), \Omega^{2}(k+1), \ldots, \Omega^{m}(k+1)$, for $k \in \mathbb{N}$.
		
		Initialization: $\Omega^{j}(0) \leftarrow 0.350^\text{1}$ , for $j \in \mathcal{M}$,
		
		\ForEach{$k \in \mathbb{N}$}{
			
			\ForEach{$j \in \mathcal{M}$}{
				
				calculate $\Omega^j(k+1)$ according to \eqref{omega_b1} and broadcast in the network;	
		} }
		\caption{Algorithm of control unit}
		\label{algoCU2}
	\end{algorithm}
	\footnotetext[1]{We initialize it with a positive real number for each resource.}
	\begin{algorithm}  \SetAlgoLined Input:
		$\Omega^{1}(k), \Omega^{2}(k), \ldots, \Omega^{m}(k)$, for $k \in \mathbb{N}$.
		
		Output: $\xi^1_i(k+1), \xi^2_i(k+1), \ldots, \xi^m_i(k+1)$, for $k \in \mathbb{N}$.
		
		Initialization: $\xi^j_i(0) \leftarrow 1$ and
		${y}^j_i(0) \leftarrow \xi^j_i(0)$, for
		$ j \in \mathcal{M}$.
		
		\ForEach{$k \in \mathbb{N} $}{
			
		\ForEach{$j \in \mathcal{M}$}{
				$\sigma^j_i(k) \leftarrow \Omega^j(k)
				\frac{{y}^j_i(k)}{ \nabla_j
					{g_i({y}_i^1(k)}, \ldots,
					{y}^m_i(k))}$; 
				
				generate Bernoulli independent random variable
				$b^j_i(k)$ with the parameter $\sigma^j_i(k)$;
				
				\eIf{ $b^j_i(k) = 1$}{
					$\xi^j_i(k+1) \leftarrow 1$; }
				{$\xi^j_i(k+1) \leftarrow 0$; }

				${y}^j_i(k+1) \leftarrow \frac{k+1}{k+2}
				{y}^j_i(k) + \frac{1}{k+2} \xi^j_i(k+1);$} 
		}
		\caption{Unit-demand multi-resource allocation algorithm of agent $i$}
		\label{algo3}
	\end{algorithm}
	After introducing the algorithms, we describe here how to calculate different factors. Let ${x}_1^1, \ldots, {x}_n^m \in [0,1]$ be the deterministic values of average allocations then the control unit calculates the \emph{gain parameter} $\tau^j$ with common knowledge of $\mathcal{G}_\delta$, for all $j$, as follows, 
	\begin{align} \label{tau}
		\tau^j \in \big ( 0,  \big ( \sup_{{x}_1^1, \ldots, {x}_n^m \in \mathbb{R}_+, g_1, \ldots, g_n \in \mathcal{G}_\delta} \sum_{i=1}^{n} \frac{ x_i^j}{\nabla_j{g_i({x}_i^1, \ldots, {x}_i^m)}} \big )^{-1} \big ).
	\end{align}
	Now, we define $\Omega^j(k + 1)$ which is based on the utilization of resource $R^j$ at time step $k$ and common knowledge $\mathcal{G}_\delta$ as follows,
	\begin{align} \label{omega_b1}
		\begin{split}
			\Omega^j(k+1) \triangleq \Omega^j(k) -  \tau^j  \Big (\sum_{i=1}^n \xi^j_i(k) -C^j \Big ).
		\end{split}
	\end{align}
	We call $\Omega^j(k)$ as the {\em normalization factor}, used by the control unit. After receiving the normalization factor $\Omega^j(k)$ from the control unit at time step $k$, agent $i$ responds with probability $\sigma_i^j(k)$ in the following manner to demand resource $R^j$ at next time step, for all $i$, $j$ and $k$:
	\begin{align} \label{prob_x2}
		\sigma_i^j(k) \triangleq  \Omega^j(k) \frac{ y^j_i(k)}{ \nabla_j{g_i({y}_i^1(k), {y}_i^2(k), \ldots, {y}_i^m(k))}}.
	\end{align}	
	Notice that $\Omega^j(k)$ is used to bound the
	probability $\sigma_i^j(k) \in (0,1)$, for all $i$, $j$ and $k$. Furthermore, note that for simplicity of notation we use $\sigma_i^j(k)$ instead of $\sigma_i^j(\Omega^j(k), y_i^j(k))$ in this section.
	%
	%
	
	Let $\xi^j(k) \in \{0,1\}^n$ and $y^j(k) \in [0,1]^n$ denote the vectors with entries $\xi_i^j(k),y_i^j(k)$, respectively, and $\sigma^j(k)$ denotes the vector with entries $\sigma_i^j(k)$, for $i=1,2,\ldots,n$, $j=1,2,\ldots,m$, and $k =0, 1, 2, \ldots$ Then similar to the single resource case, we can restate the definition of admissibility as in Definition \ref{ass:algorithm} and the theorem of convergence of average allocation $y^j(k)$ for a constant $\Omega^j$, for $j=1,2,\ldots,m$, as in Theorem \ref{t:single-res}. We state the generalized theorem of convergence of average allocations of multi-resource as follows.
\begin{theorem}[Convergence of average allocations]
	\label{t:mul-res}
	Let $n\in \mathbb{N}$. Assume that the cost function $g_i: [0,1]^n \to
	\R_+$ is strictly convex, continuously differentiable and strictly increasing in each variable, for $i=1,\ldots,n$. Furthermore, let
	$\Omega^j>0$, and assume that $\{g_i, i=1,\ldots,n \}$ and $\Omega^j$ are admissible, for $j=1,\ldots,m$. Then almost surely, $\lim_{k\to\infty} y^j(k) = y^{*j}$, where $y^{*j}$ is
	characterized by the condition,
	\begin{align} \label{eq:fixchar1}
	\Omega^j =  \nabla_j{g_i({y}_i^{*1}, {y}_i^{*2}, \ldots, {y}_i^{*m})},  \\ \text{for } i=1,\ldots,n, \text{ and } j=1,\ldots,m. \nonumber
	\end{align}
\end{theorem}
\begin{proof} We write the average allocation $y^j(k)$ as:
	\begin{align*}
	y^j(k+1) = \frac{k}{k+1} y^j(k) + \frac{1}{k+1} \xi^j(k+1),
	\end{align*}
	for $j=1,2, \ldots, m$. 
	This may be reformulated, for $j=1,\ldots,m$ as:
	\begin{align} \label{eq:avg_y_m}
	&y^j(k+1) = y^j(k) + \nonumber \\ & \hspace{0.2in}\frac{1}{k+1} \big [ \big( \sigma^j(k) - y^j(k) \big) + \big( \xi^j(k+1) - \sigma^j(k) \big) \big].
	\end{align} 
Notice that \eqref{eq:avg_y_m} is similar to \eqref{eq:avg_y}; thus, the proof follows the single resource case.
\end{proof}
Readers may note that proof of convergence with $\Omega^j(k)$ that varies with time step $k \in \mathbb{N}$ (cf. \eqref{omega_b1}), for $j=1,\ldots,m$, is an open problem.

Analogous to Remark \ref{re_Optim_s} with similar assumption on $C^{*j}$, for $j=1,\ldots,m$, we can write that the fixed point condition \eqref{eq:fixchar1} can be interpreted as an optimality condition---as established in \eqref{optimality_cond_b1}. Also, we say that $y^{*j}$, for $j=1,\ldots,m$, is the unique optimal point of the optimization Problem \ref{obj_fn1_b1}.
	\begin{remark}[Privacy of an agent]
	The control unit only knows about the aggregate utilization $\sum_{i=1}^{n} \xi_i^j(k)$ of resource $R^j$ at time step $k$ that ensures the privacy of probability and cost function of an agent. 
	\end{remark}
	Furthermore, notice that the network has very little communication overhead. Suppose that $\Omega^j(k)$ takes the floating point values represented by $\mu$ bits. If there are $m$ unit-demand resources in the network, then the communication overhead in the network will be $\mu m$ bits per time unit. Moreover, the communication complexity is independent of the number of agents participating in the
		network.
	
	We briefly present \cite{Syed2018_B} here. Let us assume that there are two unit-demand resources $R^1$ and $R^2$ in a network of $n$ agents. Agent $i$ desires to receive on-average $T_i \in [0,2]$ amount of the unit-demand resources in long-run, for $i = 1, 2, \ldots, n$. Although, the paper follows the same update scheme for normalization factors $\Omega^1(k)$ and $\Omega^2(k)$ (cf. \eqref{omega_b1}). However, the goal of the scheme is different from this paper. They aim to achieve: 
\begin{eqnarray*}
\label{eq:limit}
\lim_{k \rightarrow \infty} y^1_i(k) + y^2_i(k) & = & T_i, \quad \text{for } i = 1,\ldots, n. 
\end{eqnarray*} 
 And 
	\begin{align*}
	\lim_{k \rightarrow \infty} y^1_i(k) &= \beta_i  T_i, \quad \text{and}, 
	\lim_{k \rightarrow \infty} y^2_i(k) &= (1-\beta_i)  T_i,
	\end{align*}
	where $\beta_i \in [0,1]$, for $i = 1,2,\ldots, n$.
	\section{Application to electric vehicle charging} \label{bin_applications}
		\begin{figure*} 
		\centering
		\subfloat[]{%
			\includegraphics[width=0.29\linewidth]{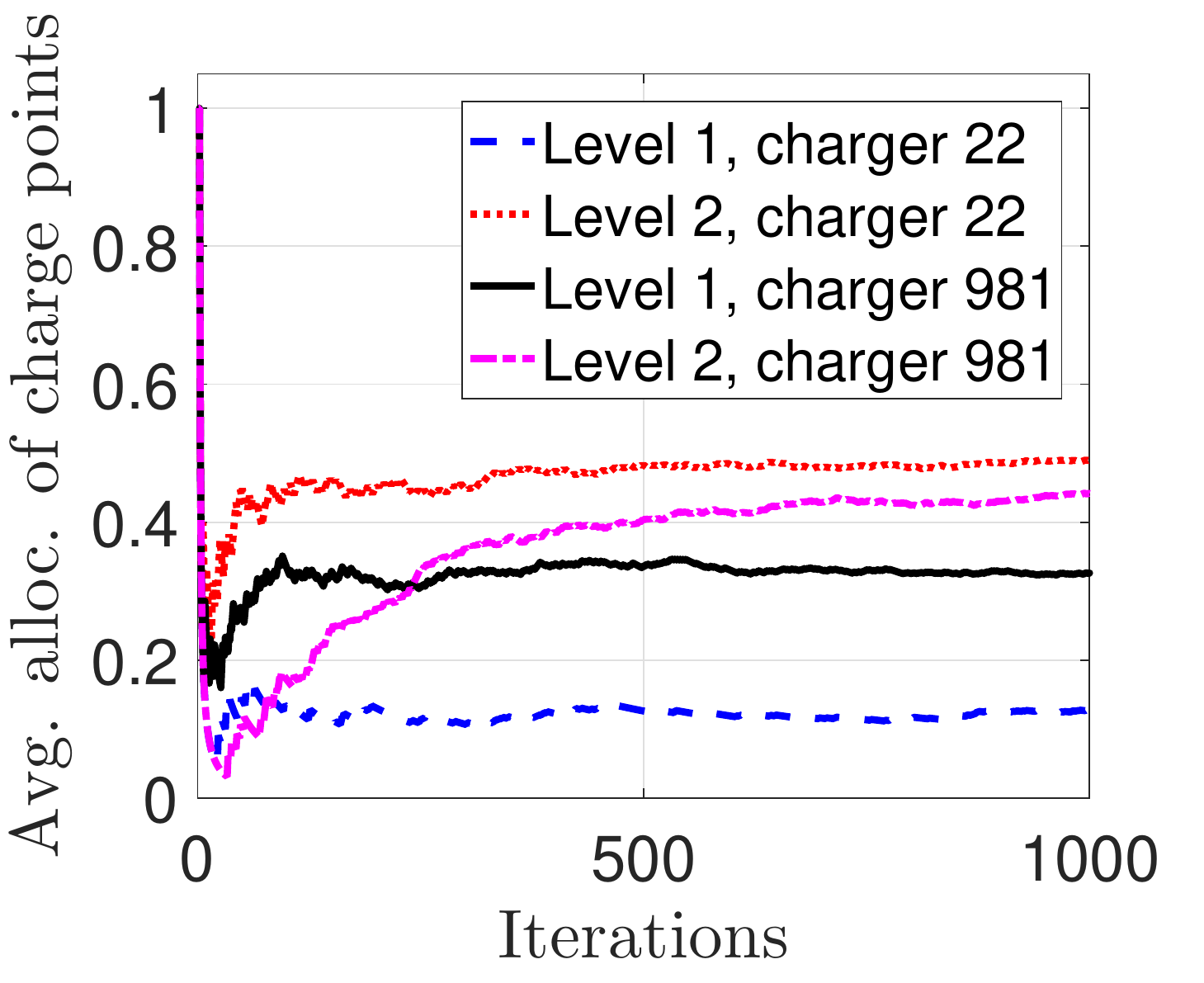}}
		\label{avg_BAIMD_2var}\hfill
		\subfloat[]{%
			\includegraphics[width=0.32\linewidth]{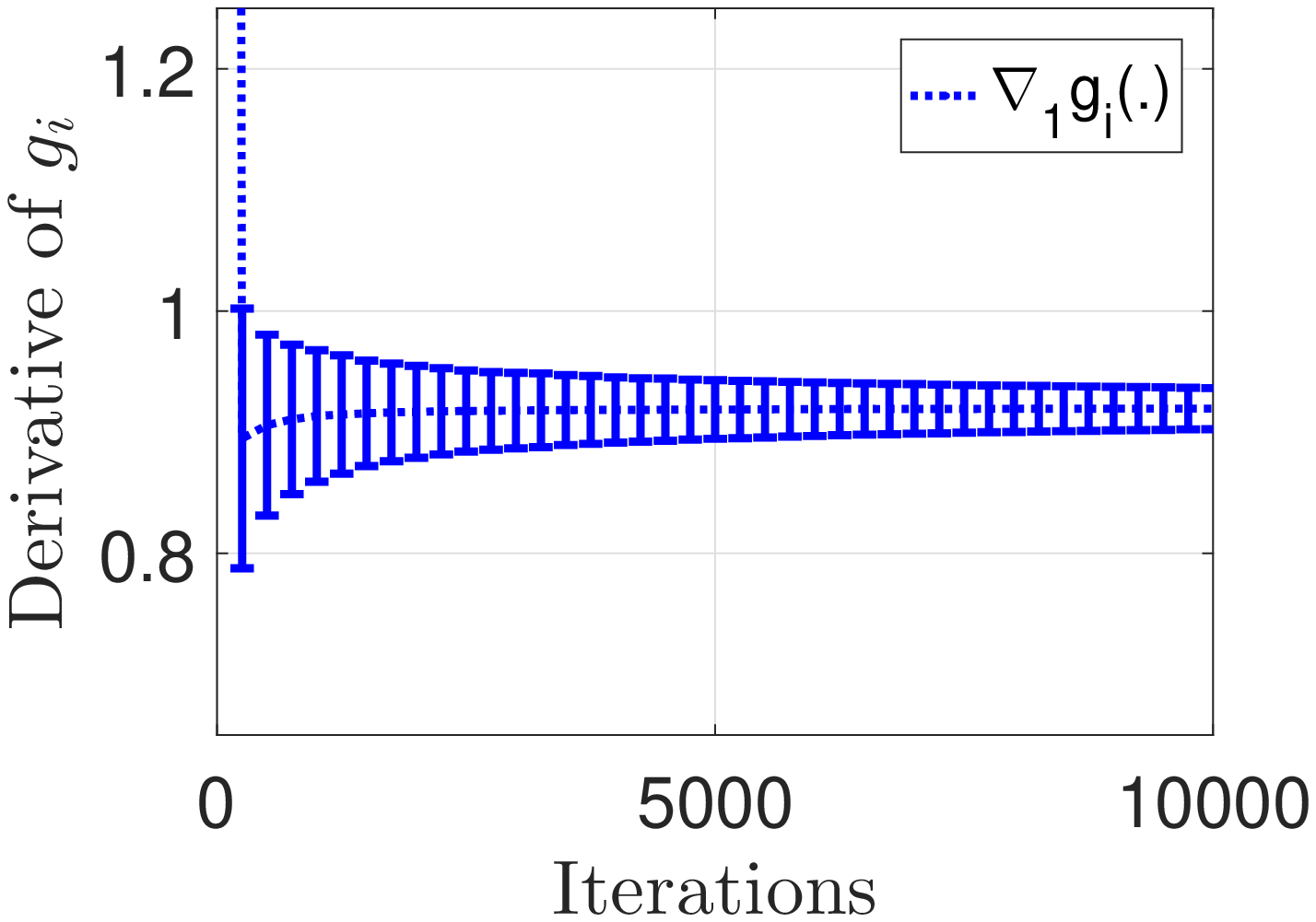}}
		\label{err_grad_BAIMD_2var_x1}\hfill
		\subfloat[]{%
			\includegraphics[width=0.32\linewidth]{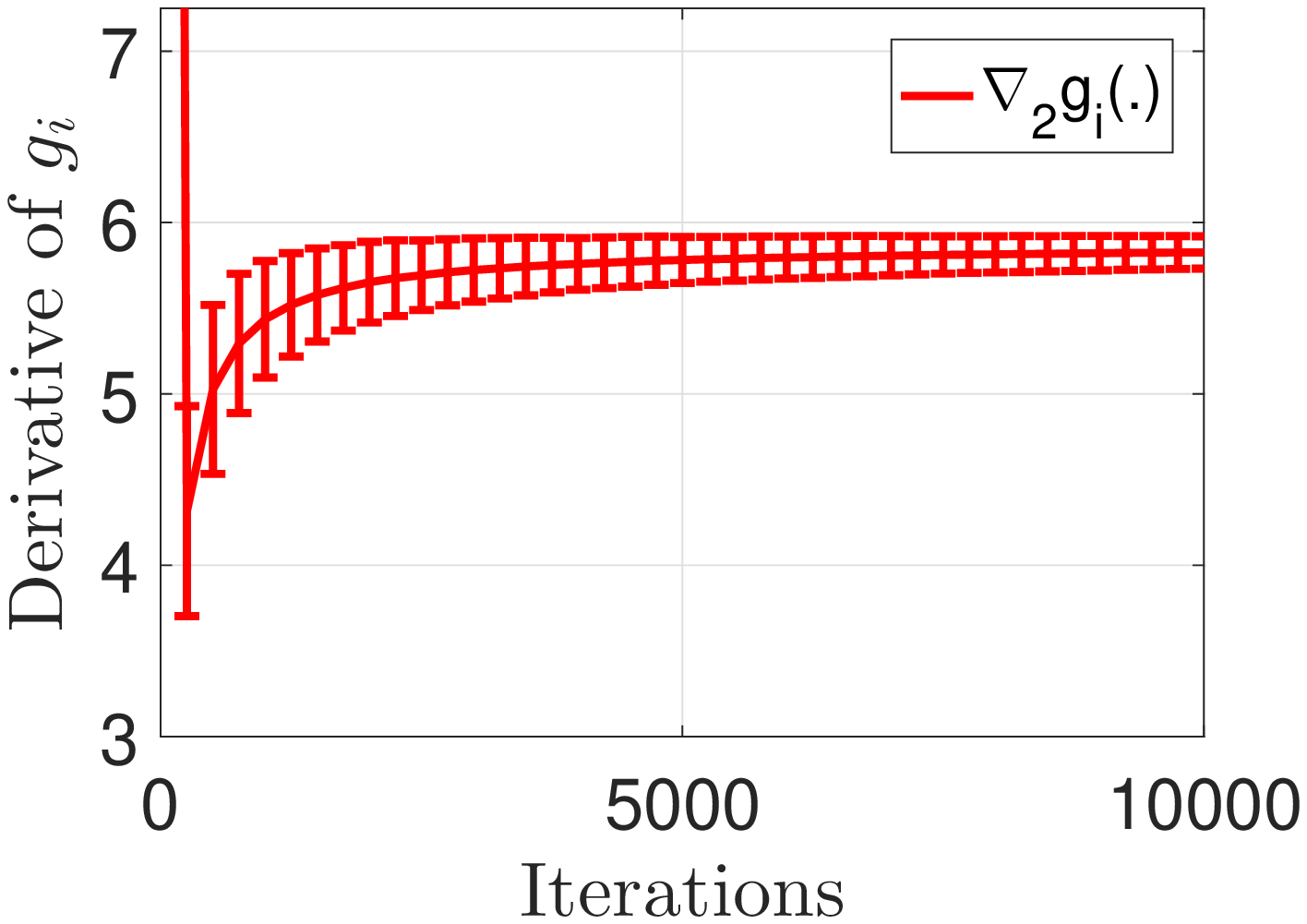}}
		\label{err_grad_BAIMD_2var_x2}\hfill
		
		\caption{(a) Evolution of average allocation of charging points, (b) evolution of profile of derivatives of $g_i$ of all the  electric cars with respect to level $1$ chargers, (c) evolution of profile of derivatives of $g_i$ of all the electric cars with respect to level $2$ chargers.}
		\label{fig3_bin} 
	\end{figure*}
	
	\begin{figure*}
		\centering
		\subfloat[]{%
			\includegraphics[width=0.32\linewidth]{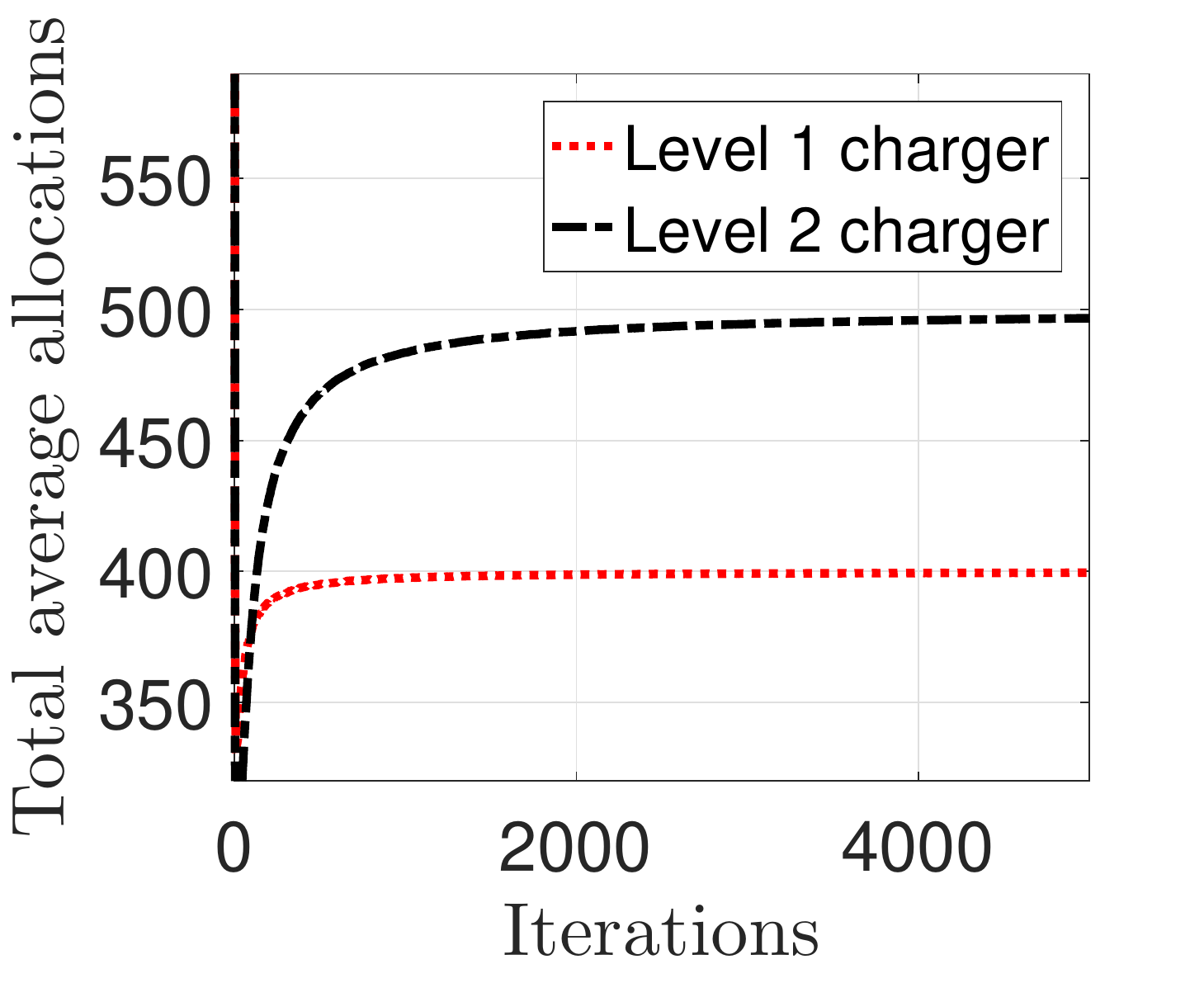}}
		\subfloat[]{%
			\includegraphics[width=0.35\linewidth]{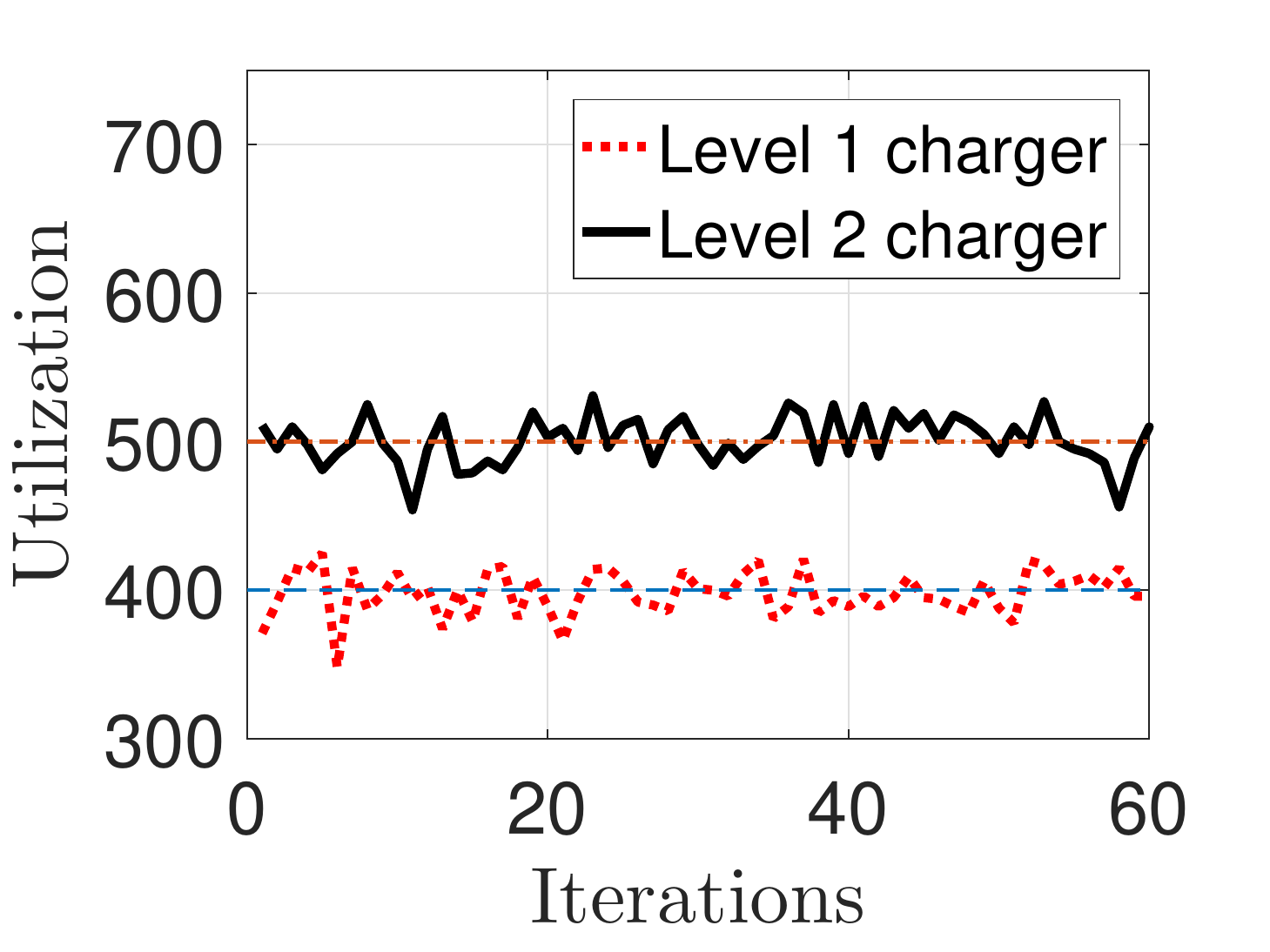}}
		\label{sum_alloc_2var}\hfill
		\captionof{figure}{
			(a) Evolution of the sum of average allocation of charging points, (b) utilization of charging points over the last $60$ time steps, capacities of level $1$ and level $2$ chargers are $C^1=400$ and $C^2=500$, respectively.}
		\label{sum_alloc_BAIMD}
	\end{figure*}
	In this section, we use Algorithms \ref{algoCU2} and \ref{algo3} to  regulate the number of electric vehicles that share a limited number of level $1$ and level $2$ charging points. We illustrate through numerical results that utilization of charging points (level $1$ or level $2$) is concentrated around its desired value (capacity); moreover, agents receive the optimal charging points in long-term averages, we verify this using the consensus of derivatives of cost functions of agents which satisfies all the KKT conditions for the optimization Problem \ref{obj_fn1_b1}, as described in Section \ref{prob_form}.
	
	As a background, the transportation sector in the US contributed around $27$\% of greenhouse gas (GHG) emission in $2015$ in which light-duty vehicles like cars have $60$\% contribution. Furthermore, the share of carbon dioxide is $96.7$\% of all GHG gases from the transportation sector \cite{USEPA2015}. To put it in context, currently, we have more than $1$ billion vehicles (electric (EV) as well as internal combustion engine (ICE)) on the road worldwide \cite{Sousanis_2018}, the number is increasing very rapidly which will result in increased  $\mathrm{CO_2}$ emission in future. Therefore, strategies are needed to reduce the $\mathrm{CO_2}$ emission. Though electric-only vehicles produce zero emission, the electricity generating units produce GHG emission at source depending on the power generation technique used, for example, thermal-electric, hydro-electric, wind power, nuclear power, etc. The US Department of Energy \cite{USDE2018} states that annual $\mathrm{CO_2}$ emission by an electric vehicle (EV) is $2,079.7$ kg (share of $\mathrm{CO_2}$ emission in producing electricity for charging an EV) and an ICE is $5,186.8$ kg.

	Now, consider a situation where a city sets aside several free (no monetary cost) {\em electric vehicle supply equipment} (EVSE) which supports level $1$ and level $2$ chargers at a public EV charging station to serve the residents or to promote usage of electric vehicles or both. Level $1$ charger works at $110\text{--}120$ Volt (V) AC, $15\text{--}20$ Ampere (A) and it takes around $8\text{--}12$ hours to charge the battery of an EV fully, whereas level $2$ charger works at $240$ V AC, $20\text{--}40$ A and it takes around $4\text{--}6$ hours to charge the battery fully---depending on the battery capacity, onboard charger capacity, and a few other factors \cite{Stephen_2014}. The voltage and current rating of chargers vary, details of ranges can be found in \cite{Yilmaz2013, Wang2016}. Furthermore, suppose that the city has installed $C^1$ EVSEs which support level $1$ chargers and $C^2$ EVSEs which support level $2$ chargers. Let $n$ electric cars are coupled through level $1$ and level $2$ charging points. Now, the city must decide whether to allocate level $1$ charging point or level $2$ charging point to an electric car to regulate the number of cars utilizing charging points. Clearly, in such a situation, charging points should be allocated in a distributed manner that preserves the privacy of individual car users, but also maximizes the benefit to the city. We use the proposed distributed stochastic algorithm which ensures the privacy of electric car users and allocates charging points optimally to maximize social welfare, for example, to minimize total electricity cost or $\mathrm{CO_2}$ emission. 
	
	According to \cite{EuroAss_2018}, on average $0.443$ kg of $\mathrm{CO_2}$ is produced to generate and distribute $1$ kWh of electric energy in the European Union with mix energy sources. Let $I$ be the current flowing in the circuit and $V$ be the voltage rating of the circuit, let $E_{CO_{2}}$ be the rate of $CO_2$ emission per kWh. If an EV is charged for $t$ hours at a charging point then its total share of $\mathrm{CO_2}$ emission, say $T_{\mathrm{CO_2}}(t)$ for generation and distribution of $I \times V \times t$ kWh electric energy is calculated as $T_{\mathrm{CO_2}}(t) = I \times V \times t \times E_{\mathrm{CO_2}}$, we use $E_{\mathrm{CO_2}} = 0.443$ kg.   
	 Table \ref{Co2_tab} illustrates the total $\mathrm{CO_2}$ emission in kg by level $1$ and level $2$ chargers in four-hours duration. We use this data to formulate the cost function  $g_i$ of (electric) car $i$, for all the cars.
%
\begin{table}[h]
\centering
\begin{tabular}{||c c c||} 
 \hline
 Charger type & power (kW) & $\mathrm{CO_2}$ emission in four hours \\
 \hline\hline
 Level $1$ & $1.65\text{--}2.40$ & $2.92\text{--}4.25$ kg \\ 
 Level $2$ & $4.80\text{--}9.60$ & $8.51\text{--}17.01$ kg \\
 \hline
\end{tabular}
\caption{$\mathrm{CO_2}$ emission in generation and distribution of electricity}
\label{Co2_tab}
\end{table}		

	Suppose that each car user has private cost function $g_i$ which depends on the average allocations $y_i^1(k)$ and $y_i^2(k)$ of level $1$ and level $2$ charging points, respectively, for $i=1,2,\ldots,n$. We assume that the city agency (control unit) broadcasts the normalization factors $\Omega^1(k)$ and $\Omega^2(k)$ to each competing electric car after every $4$ hours, here we chose a duration of $4$ hours because of charging rate of level $2$ chargers. Note that an EV user can unplug the vehicle in the middle of charging without fully charging the battery. Now, suppose that the cost functions are classified into four classes based on---the type of vehicle, its battery capacity, onboard charger capacity, and a few other factors. We assume that a set of vehicles belonging to each class. Based on the values in Table \ref{Co2_tab}, we let  the constants $a=2.9$, $b=8.51$, and let $f_{1i}$, $f_{2i}$ be uniformly distributed random variables, where $f_{1i} \in [1, 1.5]$, $f_{2i} \in [1, 2]$, for all $i$. The cost function $g_i$ is listed in \eqref{bin_func}, where first and second terms represent $\mathrm{CO_2}$ emission at a basic assumed rate of charging of the battery, whereas third and subsequent terms are $\mathrm{CO_2}$ emission due to different charging losses or factors. We observe that no allocation of charging points produce zero $\mathrm{CO_2}$ emission, the cost functions are as follows,  
		\begin{equation} \label{bin_func} g_{i}(y_i^1,y_i^2)= \left\{
		\begin{array}{ll}
		(i) \text{  } a y_i^1 + b y_i^2 + af_{1i}(y_i^1)^2 + bf_{2i}(y_i^2)^4, \\
		(ii) \text{  } a y_i^1 + b y_i^2 + af_{1i}(y_i^1)^4/2 + bf_{2i}(y_i^2)^2, \\
		(iii) \text{  } a y_i^1 + b y_i^2 +  af_{1i}(y_i^1)^4/3 + af_{1i}(y_i^1)^6 +&\\ \quad \quad bf_{2i}(y_i^2)^4, \\
		(iv) \text{  } a y_i^1 + b y_i^2 + af_{1i}(y_i^1)^2 + bf_{2i}(y_i^2)^6.
		\end{array}
		\right.
		\end{equation}
		 Now, let the number of electric cars be $n=1200$ that use level $1$ and level $2$ chargers. We classify these cars as follows---cars $1$ to $300$ belong to class $1$, cars $301$ to $600$ belong to class $2$, cars $601$ to $900$ belong to class $3$, and cars $901$ to $1200$ belong to class $4$. Each class has a set of cost functions; the cost functions of class $1$ are presented in \eqref{bin_func}$(i)$ and analogously for other classes. Let $C^1 = 400$ and $C^2 = 500$. The parameters of the algorithms are initialized with the following values; $\Omega^1(0) = 0.328 $, $\Omega^2(0) = 0.35$, $\tau^1 = 0.0002275$, and $\tau^2= 0.0002125$. We use the proposed Algorithm \ref{algoCU2} and Algorithm \ref{algo3} to allocate charging points to $n$ electric cars that are coupled through level $1$ and level $2$ charging points. If a car user is looking for a free charging point, then it sends a request to the city agency in a probabilistic manner based on its private cost function $g_i$ and its previous average allocation of level $1$ and level $2$ charging points. Based on the request, the city agency allocates one of the charging points or both or none. Furthermore, the car users do not share their cost functions or history of their allocations with other car users or with the city agency. Notice a limitation of this application, following the proposed algorithm, in some cases; a car user can receive access to both level $1$ and level $2$ charging points for a single car, which may not be desired in real-life scenarios.
	 			  	
	We present simulation results of automatic allocation of charging points here. We observe that the electric car users receive optimal allocations of both types of charging points and minimize the overall $\mathrm{CO_2}$ emission. Moreover, we observe in Figure \ref{fig3_bin}(a) that the long-term average allocations of charging points of electric cars converge to their respective optimal values.

	As described earlier in \eqref{optimality_cond_b1}, to show the optimality of the solution, the derivatives of the cost functions of all the cars with respect to a particular type of charger should make a consensus. The profile of derivatives of cost functions of the cars with respect to level $1$ and level $2$ chargers for a single simulation is illustrated in Figure \ref{fig3_bin}(b) and \ref{fig3_bin}(c), respectively. We observe that they converge with time and hence make a consensus, which meets the KKT conditions for optimality. Note that we use third and subsequent terms of \eqref{bin_func} to calculate the derivative $\nabla_j g_i$ which shifts its value by constants $a$ or $b$ without affecting the KKT points, but it provides faster convergence in the simulation. The empirical results thus obtained, show the convergence of the long-term average allocations of charging points to their respective optimal values using the consensus of derivatives of the cost functions, which results in the optimum emission of $\mathrm{CO_2}$. We also observed that $\sigma_i^j(k)$ is in $(0,1)$ most of the time with the current values of $\Omega^j(k)$ and $\tau^j$ with a few initial overshoots. To overcome the overshoots of probability $\sigma_i^j(k)$, we use $\sigma_i^j(k) = \min \Big\{\Omega^j(k) \frac{ y^j_i(k)}{ \nabla_j{g_i({y}_i^1(k), {y}_i^2(k), \ldots, {y}_i^m(k))}}, 1 \Big\}$, for all $i, j$ and $k$.

	Figure \ref{sum_alloc_BAIMD}(a) illustrates the sum of the average allocations $\sum_{i=1}^{n} {y}_i^j(k)$ over time. We observe that the sum of the average allocations of charging points converge to respective capacity over time that is, for large $k$, $\sum_{i=1}^{n} {y}_i^j(k) \approx C^j$, for all $j$. We further illustrate the utilization of charging points for the last $60$ time steps in Figure \ref{sum_alloc_BAIMD}(b). It is observed that most of the time the total allocation of charging points is concentrated around its capacity. To reduce the overshoot of total allocation of level $j$ charging points, we assume a constant $\gamma^j <1$ and modify the algorithm of the city agency (control unit) to calculate $\Omega^j(k+1)$ (cf. \eqref{omega_b1}) in the following manner,
	\begin{align*}
		\Omega^j(k+1) = \Omega^j(k) - \tau^j \Big
		(\sum_{i=1}^n \xi^j_i(k) -\gamma^jC^j \Big ),
	\end{align*}
	for $j =1, 2$ and all $k$.

	\section{Conclusion} \label{conc} We proposed a new algorithm to solve a class of multi-variate resource allocation problems. The solution approach is distributed among the agents and requires no communication between agents and little communication with a central agent. Each agent can, therefore, keep its cost function private. This generalizes the unit-demand single resource allocation algorithm of \cite{Griggs2016}. In the single-resource case, based on a constant normalization factor, we showed that the long-term average allocations of a unit-demand resource converge to optimal values; multiple (unit-demand) resource case follows this result. Additionally, experiments show that the long-term average allocations converge rapidly to optimum values in the multi-resource case.

Open problems are to prove convergence with a time-varying normalization factor for single-resource as well as multi-resource cases. Another open problem is to analyze the rate of convergence. In terms of applications, our proposed approach can be used to allocate resources, such as Internet-of-Things (IoT) devices in hospitals, smart grids, to list a few. It can also be used to allocate virtual machines to users in cloud computing. 

	\bibliographystyle{IEEEtran}
	\bibliography{DistOpt_bib} 

\end{document}